\documentclass[12pt]{article}
\usepackage[margin=1in]{geometry}
\usepackage{amsmath,amssymb,amsthm,mathrsfs}
\usepackage{hyperref}
\usepackage{xurl}
\usepackage{natbib}

\providecommand{\APACmonth}[1]{}
\providecommand{\APACday}[1]{}

\usepackage{graphicx}
\usepackage{caption}
\usepackage{enumitem}
\usepackage{mathtools}
\usepackage{tikz}
\usetikzlibrary{arrows.meta,calc,decorations.pathmorphing}

\newcommand{\be}{\begin{equation}}
\newcommand{\ee}{\end{equation}}

\newcommand{\owedge}{\mathbin{\bigcirc\mspace{-15mu}\wedge\mspace{3mu}}}

\newtheorem{prop}{Proposition}
\newtheorem{thm}{Theorem}
\newtheorem{lemma}{Lemma}

\theoremstyle{remark}
\newtheorem{remark}{Remark}

\title{$\delta$-flatness and its natural extension}
\author{Henrique de A. Gomes \\ \it University of Bonn, Bonn \\ \it Oriel College, University of Oxford \\ \tt gomes.ha@gmail.com}
\date{}

\begin{document}

\maketitle

\begin{abstract}
What, if anything, makes Minkowski spacetime a privileged local reference geometry in General Relativity? Recent work by \cite{WeatherallFletcher} argues: nothing. In response, I proposed a criterion of ``$\delta$-flatness'', which bounds the magnitude of tidal acceleration within a tubular neighborhood by $\delta$, and showed that Minkowski uniquely saturates the bound. Here I ask how $\delta$-flatness generalises, and distinguish two kinds of extension. Reference-metric generalisations compare physical tidal acceleration with a reference tidal-force term constructed from another metric. Lorentzian signature obstructs every such comparison in which the reference fails to share the lightcones of the physical metric. Criterion modifications instead build the criterion from the physical geometry alone. The space of such modifications is wide, and I focus on those closest to $\delta$-flatness, where one bounds the deviation of tidal acceleration from a non-zero target. The minimal candidate, which I call $\delta$-MSS, has a singular limit that picks out the maximally symmetric spacetimes (Minkowski, de Sitter, Anti-de Sitter) via Schur's theorem. A restricted variant, which compares $g$ with a conformal rescaling of itself, either collapses to $\delta$-MSS or fails to pick out a unique geometry. The chief result is that $\delta$-flatness distinguishes Minkowski with no free parameter. $\delta$-MSS, a weaker variant, picks out the maximally symmetric family at the cost of an externally supplied scalar. %It is possible that criteria built from different mathematical machinery might single out other spacetimes, but pursuing this is beyond the scope of this paper.
\end{abstract}

\section{Introduction}

The question of whether general-relativistic spacetimes are ``locally flat'' has recently attracted renewed philosophical attention. Einstein's Equivalence Principle suggests that the dynamics in a sufficiently small region of a curved spacetime should be indistinguishable from those in special relativity. If Special Relativity is merely one solution among many in the ``panoply'' of General Relativity's models, its role is pragmatic rather than foundational. But if Special Relativity holds \emph{locally} in a substantive sense, so that curved spacetimes are ``locally special relativistic'', then special relativistic reasoning is not merely a convenient approximation but a licensed guide to qualitative general relativistic phenomena. Much of what is qualitatively distinctive about GR turns on geodesic deviation.

In recent work, \cite{WeatherallFletcher} have argued that standard formulations of local flatness fail to provide a substantive sense in which flat spacetime plays a privileged role in General Relativity. One class of formulations appeals to coordinate constructions, notably the existence of Riemann normal coordinates in which the metric components are Minkowskian and the connection coefficients vanish at a point. But a central lesson of General Relativity is that physical effects are geometrical and coordinate-independent. Privileging flat spacetime via a special class of coordinate systems runs against the spirit of the theory, and at a minimum invites scrutiny.

 A second class of formulations employs the isomorphism between the tangent space at any point and Minkowski spacetime. As \cite{WeatherallFletcher} correctly observe, such formulations do not involve any approximation: the tangent space at every point of every Lorentzian manifold is Minkowski spacetime, exactly. Such a property is strictly differential-geometric: it is not a relation of approximation between the curved spacetime and flat spacetime, and so  it  cannot ground any claim that a spacetime is ``locally \emph{approximately} flat.''

  The central result of \cite{WeatherallFletcher} (Theorem 5) sharpens these ideas: they show that to first order in metric derivatives, every Lorentzian metric locally approximates every other Lorentzian metric. Local flatness, on any of these readings, is not a substantive geometric property.

In response, \cite{Teh2024} offered two defenses. The first ties Riemann normal coordinates to geodesic deviation, arguing that these coordinates have privileged physical status. But this still operates within the framework of coordinate expansions. The second describes the metric expansion purely geometrically, using Synge's world function and related bitensor technology, and is genuinely coordinate-free. \citet{LinnemannReadPuzzle}, in a recent appraisal of the debate, have observed that the geometric reformulation does not  dissolve the base-point problem: even the Synge-style construction must select a Minkowski metric around which to expand. The expansion
 takes the  form
\begin{equation}
\label{eq:LR-expansion}
g_{IJ} = \eta_{IJ} - \tfrac{\delta^2}{3}\,R_{ILJK}\,x^L x^K + \mathcal{O}(x^3),
\end{equation}
with $\eta_{IJ}$ as the chosen Minkowski base point and curvature entering through higher-order corrections (developed at length in \citealp{Teh2024}). Here $\delta$ is the ratio of a probing length to a characteristic curvature length. Linnemann and Read note that the choice of Minkowski is not uniquely constrained by the expansion. The expansion does not fix which particular Minkowski metric is intended, and indeed it does not require that Minkowski rather than some other metric should be the base point at all: this is described by Hari--Kothawala (\citeyear{HariKothawala2020}), who show how generalised normal coordinates can describe a similar expansion around an arbitrary reference Lorentzian metric.

In previous work \citep{Gomes_flat}, I proposed a different defense that bypasses the coordinate question entirely. The geodesic deviation equation \citep{HawkingEllis, Poisson_book} gives the relative acceleration between neighboring geodesics. Such a quantity  is both coordinate-invariant and operationally measurable. Clearly, it requires neither a choice of a chart nor of a frame. A spacetime is then called ``$\delta$-flat'' if the norm of this acceleration along geodesic segments remains bounded by $\delta$ within a tubular neighborhood of fixed proper radius $\rho > 0$. While every spacetime is $\delta$-flat for some such $\rho>0$, it is only Minkowski spacetime that saturates this bound: as $\delta \to 0$,  for $\rho>0$,  the Riemann tensor must vanish.

$\delta$-flatness is coordinate-invariant, but one might still ask whether it tacitly privileges Minkowski as a base point for a functional expansion. But this challenge---that \citet{LinnemannReadPuzzle} press against expansion-based approaches (``which Minkowski?'', ``why Minkowski rather than another metric?'', their concerns following \eqref{eq:LR-expansion})---has no purchase here, as the flat metric is not at all invoked in the definition. The $\delta$-flatness criterion  picks no metric base point: it bounds an intrinsic quantity (the magnitude of tidal acceleration). And the bound is saturated by flat spacetime because the absence of tidal force is an intrinsic property of a spacetime, not a comparison with any reference.\footnote{ One should not be tempted by the algebraic identity $\|\text{Dev}(g)\| = \|\text{Dev}(g) - \text{Dev}(\eta)\|$ into reading the criterion as covertly comparing $g$ to $\eta$: the same identity equally motivates a generalisation of the form $\|\text{Dev}(g) - P(\text{Dev}(g'))\|$ for any polynomial $P$ with $P(0)=0$, each such $P$ recovering $\delta$-flatness in the case $g'=\eta$. The rewriting under-determines the generalisation.}

A separate question can still arise: whether $\delta$-flatness has a natural extension that picks out non-Minkowski reference geometries. Two kinds of extension present themselves, and disentangling them is the central business of this paper.

The first kind of extension \emph{compares a given metric with a reference metric}: such a criterion would bound not the magnitude of tidal acceleration itself but its deviation from a reference tidal-force term constructed from another metric $g'$. It would ask whether $(M,g)$ is $\delta$-similar to $(M,g')$. Could one define a ``$\delta$-Schwarzschild'' condition bounding the difference between the tidal forces in a laboratory spacetime and those of a black hole? Section~\ref{sec:obstruction} shows that Lorentzian signature obstructs this kind of extension whenever $g'$ doesn't share the lightcones of $g$. The mismatch between physical and reference geometry destroys the orthogonality on which $\delta$-flatness depends, and opens a ``null degeneracy'' that lets infinite classes of distinct spacetimes saturate the bound.

The second kind of extension \emph{modifies the criterion itself}: change the right-hand side, or more generally the entire content of the inequality, drawing only on the geometry of the physical spacetime, without invoking any reference metric. This amounts to choosing something other than solely the bare tidal force to construct a criterion of local approximations of general spacetimes.

The space of such admissible  modifications is wide and my interests here, which I will defend below, are on criteria based on tidal acceleration, so  I will investigate only  the corner closest to $\delta$-flatness, where one bounds the deviation of tidal acceleration from a non-zero target. Within that corner the minimal nontrivial choice is $\delta$-MSS (\S\ref{sec:homogeneous}). The construction preserves the orthogonality on which $\delta$-flatness depends, so the null degeneracy does not arise. Its singular limit forces the Riemann tensor into the constant-sectional-curvature form, and Schur's theorem constrains that to be a global constant. The saturating geometries are exactly the maximally symmetric spacetimes (MSS): Minkowski, de Sitter, and Anti-de Sitter, one for each sign of the constant. $\delta$-MSS singles out the maximally symmetric family only once one supplies $\Lambda$ (independently of the apparatus). One restricted construction survives the obstruction of \S\ref{sec:obstruction}: comparing $g$ with a rescaling $\Omega^2 g$ of itself, which preserves the lightcones. Since the comparator varies with $g$ once $\Omega$ is fixed, this is a modification of the criterion rather than a comparison with a fixed reference, and I treat it in \S\ref{sec:homogeneous}, alongside $\delta$-MSS. It either collapses to $\delta$-MSS or admits non-unique saturating families (Appendix~\ref{app:conformal-saturation}).

Whether other operationally motivated criteria, drawn from different mathematical machinery, could single out other reference geometries is a substantive question. It is not easy to `reverse-engineer' the operational criterion that would single out a given geometry, and pursuing the question would go well beyond the scope of this paper.

\paragraph{A word on scope.}

The question of local similarity can be posed in many different senses, only one of which is the focus of the present paper.  I focus on the \emph{tidal-acceleration} sense: which extensions of $\delta$-flatness yield well-posed local criteria, whether by comparison with a reference metric or by modification of the criterion itself.

It is easy to defend tidal-acceleration as a natural choice. For, to the extent that it operationally encodes the Riemann curvature, it is, as \cite{Synge1960} insisted, the only observer-independent meaning the term ``gravitational field'' admits: ``In Einstein's theory, either there is a gravitational field or there is none, according as the Riemann tensor does not or does vanish. This is an absolute property'' (\citealp{Synge1960}, p.~ix). The geodesic deviation equation is the local instantiation of that curvature, and many canonical applications of GR employ it. For instance, gravitational-wave detection measures tidal strain directly; gravitational lensing and the bending of starlight measure the focusing of light rays by curvature; the singularity theorems use the geodesic deviation equation as their main technical tool, etc. The $\delta$-flatness criterion exploits this directly by bounding the magnitude of tidal force, with no reference to coordinates or frames.

In sum, the criterion under which uniqueness holds employs the geodesic-deviation framework, which I take to be the natural setting given the desiderata of coordinate-invariance and operational accessibility. Thus what I present is one robust sense in which flat spacetime uniquely approximates a general spacetime. Other formulations, such as first-order metric expansions, Synge bitensor expansions, or frame-bundle approximations, may provide different senses of local similarity, which I will not try to explore. 

\paragraph{A second word on scope: chronogeometry.}\label{par:chronogeometry}
The criterion also takes for granted the standard chronogeometric interpretation of $g$: free-falling test bodies follow $g$-geodesics, and the clocks and rods of an inertial observer record $g$-intervals. That reading is common ground in the debate relevant for this paper. The question at issue is whether any substantive local-approximation property survives the analysis of \citet{WeatherallFletcher}; whether and how $g$ acquires chronogeometric significance is beyond the scope of the paper. Nonetheless, I admit that this assumption limits the reach of the criterion that I am proposing. Some discussions of local flatness aim at chronogeometry itself, asking in virtue of what the metric has the physical meaning that it does, as in the dynamical approach to spacetime theories \citep{Brown_book,BrownRead_dyn}. %In that setting one cannot argue that $g$ inherits its significance from the local approximate validity of Special Relativity and then cash out the latter as $\delta$-flatness. The operational reading of the criterion already presupposes what such an argument sets out to establish, and the reasoning would move in a circle. 
Where I appeal below to what clocks and rods measure (\S\ref{sec:obstruction}), I presuppose this background.

\section[The Mechanism of delta-Flatness]{The Mechanism of \texorpdfstring{$\delta$}{delta}-Flatness}\label{sec:mechanism}

Two astronauts in free fall, initially at rest relative to each other, will drift apart (or together) if spacetime is curved. This tidal acceleration, the relative acceleration of neighboring geodesics, gives the operational meaning of curvature, and it is coordinate-invariant: no choice of chart can make it vanish where the Riemann tensor does not vanish.

Now we will see that $\delta$-flatness provides an interesting unique bound because tidal acceleration is orthogonal to the observer's four-velocity. This orthogonality confines the acceleration to the spacelike subspace, where the Lorentzian norm is positive-definite, ensuring the relevant inequality is uniquely saturated.

Consider a timelike geodesic $\gamma$ with unit tangent $v^a$ ($v^a v_a = -1$, $v^a \nabla_a v^b = 0$). Let $r^a$ be a deviation vector connecting $\gamma$ to a neighboring geodesic, such that the Lie bracket vanishes: $[v, r] = v^a\nabla_a r^b - r^a\nabla_a v^b = 0$. We can choose $r^a$ orthogonal to $v^a$ at a point, and this condition propagates along the geodesic.

The relative acceleration is given by the geodesic deviation equation:
\begin{equation}
\label{eq:geodesic_dev}
(\text{Dev}(g)\cdot v)^a := v^c\nabla_c(v^b\nabla_b r^a) = {R^a}_{bcd} v^b v^d r^c.
\end{equation}
From now we will omit the dependence on $v$ and write $\text{Dev}(g)^a$.  Contracting this deviation with $v_a$:
\begin{equation}
v_a \text{Dev}(g)^a = R_{ebcd} v^e v^b v^d r^c = 0
\end{equation}
by the antisymmetry $R_{ebcd} = -R_{becd}$.

This orthogonality confines the tidal acceleration to the spacelike subspace orthogonal to $v^a$, where the Lorentzian metric induces a positive-definite norm:
\begin{equation}
\|\text{Dev}(g)\| := \sqrt{g_{ab} \,\text{Dev}(g)^a \,\text{Dev}(g)^b}
\end{equation}

A spacetime is \emph{$\delta$-flat} if every timelike geodesic can be covered by segments for which this magnitude remains bounded by $\delta$ within a tubular neighborhood of fixed proper radius $\rho > 0$:
\begin{equation}
\label{eq:delta_flat}
\|\text{Dev}(g)\| < \delta
\end{equation}
In the limit $\delta \to 0$ (for any fixed $\rho>0$), the condition $\|\text{Dev}(g)\| \to 0$ implies $\text{Dev}(g)^a \to 0$, because for spacelike vectors $\|X\| = 0$ only if $X^a = 0$. Since this must hold for arbitrary $r^c$ and $v^b$, the Riemann tensor itself must vanish:\footnote{The implication ``$\text{Dev}(g)^a=0$ for all timelike $v^b$ and orthogonal $r^c \Rightarrow R^a{}_{bcd}=0$'' is the standard fact that the sectional curvature on two-planes determines the Riemann tensor in any signature, given its algebraic symmetries; see e.g.\ \citet[Prop.~41, p.~79]{Oneill} or \citet[\S 3.2b]{Wald_book}.}
\begin{equation}
{R^a}_{bcd} = 0 \quad \text{(Flat Spacetime)}
\end{equation}
Minkowski spacetime uniquely saturates the condition: smaller errors correspond to smaller curvature, and vanishing error to flat geometry.

\section[Reference-Metric Generalisations]{Reference-Metric Generalisations}\label{sec:obstruction}

$\delta$-flatness bounds an intrinsic quantity and need not be formulated as a comparison to an auxiliary metric.  But the question of how it generalises is independently substantive, and two kinds of generalisation present themselves. One kind replaces the $0$ on the right-hand side of \eqref{eq:delta_flat} with a reference tidal-force term constructed from another metric $g'$. The other modifies the right-hand side directly, or more generally the entire content of the inequality, without invoking any reference metric. This section takes up the first kind and the second is the subject of \S\ref{sec:homogeneous}.

A natural proposal in the first family would define ``$\delta$-similarity'' to an arbitrary reference spacetime $(M, g')$, characterised by its own Riemann tensor $R'_{abcd}$. If well-posed for arbitrary $g'$, the result would be a family of local comparisons indexed by a reference geometry, with $\delta$-flatness as the case $g'=\eta$. As we shall see, Lorentzian signature obstructs the generalisation for every reference that doesn't share the lightcones of $g$. Restricting attention to physical metrics that share the lightcones of the reference yields a construction indexed by a conformal factor rather than by a fixed reference metric. I treat it in \S\ref{sec:homogeneous} as a criterion modification.

\subsection{The Null Degeneracy}

How should we define the ``reference tidal force'' from $g'$? Again, two proposals present themselves. 

\paragraph{The dynamical proposal.}
The direct option is to use relative acceleration using the reference connection $\nabla'$, writing down the geodesic deviation in $(M,g')$. But a laboratory observer follows a $g$-geodesic, not a $g'$-geodesic. In the reference geometry, the same curve is generally accelerating, with proper $g'$-acceleration $a'^a = C^a_{bc}\,v^b v^c$, where $C^a_{bc} := {\Gamma'}^a_{bc} - \Gamma^a_{bc}$ is the difference connection. Along a curve that is not a $g'$-geodesic, the deviation equation acquires  terms such as:
\begin{equation}
(\text{Dev}(g')\cdot v)^a
= {R'}^a_{\ bcd}\,v^b v^d r^c
+ \underbrace{(\nabla'_{r} C^a_{de})\,v^d v^e
+ 2\,C^a_{de}\,(\nabla'_{r} v^d)\,v^e}_{\text{inertial obstruction }\mathcal{O}^a}.
\end{equation}
The obstruction terms $\mathcal{O}^a$ are not features of $g'$'s curvature. They arise when we use $\nabla'$ on a curve that is geodesic with respect to $\nabla$. Of course, what we want to ask is how the tidal forces in $(M,g)$ compare with those $g'$ \emph{would} produce, but $(\text{Dev}(g')$ has terms that are artifacts of using a fictitious metric that produces fictiotious inertial terms.

\paragraph{The algebraic curvature.}
A natural alternative is to keep only the curvature term:
\begin{equation}
\label{eq:ref-alg}
\text{Ref}(g')^a := {R'}^a_{\ bcd}\,v^b v^d r^c.
\end{equation}
I will write ``Ref'' rather than ``$\text{Dev}(g')$'' to emphasise that this is not the geodesic-deviation acceleration of $(M,g')$ but a reference tidal-force term constructed from $g'$ and evaluated in the physical spacetime.

Here such a $\text{Ref}(g')^a$ is a hybrid. The metric $g$ supplies the kinematics, namely the curve $\gamma$ as a geodesic with four-velocity $v^a$ and part of a congruence with deviation vector $r^a$, while $g'$ supplies only the curvature. A $g$-geodesic is generally not a $g'$-geodesic, so $v^a$ need not be $g'$-timelike, and $r^a$ need not be $g'$-orthogonal to $v^a$. ``$\delta$-similarity to $g'$'' is therefore here shorthand for ``$\delta$-closeness in tidal-force, with $\text{Ref}$ drawn from $R'$''; it is a comparison between two Riemann tensors evaluated on common $g$-kinematic objects. 

The proposal is then to call $(M,g)$ \emph{$\delta$-similar to $g'$} if
\begin{equation}
\label{eq:delta_sim}
\|\, \text{Dev}(g)^a - \text{Ref}(g')^a \,\| < \delta.
\end{equation}
For this to work, that is, for $\delta\to 0$ to imply geometric identity to $g'$, the difference vector must be strictly spacelike. 

For the physical tidal force, the analogous contraction of the Riemann tensor with the deviation vector vanishes because $R_{abcd}=-R_{bacd}$. 
But contracting $\text{Ref}(g')^a$ with $v_a$ we find:
\begin{equation}\label{eq:contraction}
v_a\,\text{Ref}(g')^a =v^e\, g_{ae}\,{R'}^a_{\ bcd}\,v^b v^d r^c=v^e\,g_{ae}\,g'^{fa}\,{R'}_{fbcd}\,v^b v^d r^c.
\end{equation}
Here the object $g_{ae}{R'}^a_{\ bcd}$ is \emph{not} $R'_{ebcd}=g'_{ae}{R'}^a_{\ bcd}$: the antisymmetry $R'_{ebcd}=-R'_{becd}$ holds for indices lowered by $g'$, not by $g$. The mixed tensor $g_{ae}{R'}^a_{\ bcd}$ has no particular symmetry in $e$ and $b$, and its contraction with the symmetric $v^e v^b$ need not vanish. In other words, in \eqref{eq:contraction} we have $g_{ae}\,g'^{fa}\neq \delta_e^f$. So
\begin{equation}
v_a\,\text{Ref}(g')^a \neq 0.
\end{equation}
For some observers, then, the reference force acquires a timelike component, and the difference vector $D^a = \text{Dev}(g)^a - \text{Ref}(g')^a$ lives in the full Lorentzian tangent space. This is fatal: in Lorentzian geometry, a non-zero vector can have vanishing norm. If $D^a$ is null,
\begin{equation}
\|D\|^2 = g_{ab}D^a D^b = 0
\quad\text{even though}\quad
D^a\neq 0.
\end{equation}
The $\delta$-bound \eqref{eq:delta_sim} is then satisfied as $\delta\to 0$ by an infinite class of spacetimes whose tidal forces differ from those of $g'$ by a null vector, including certain gravitational-wave perturbations. The criterion characterises an entire equivalence class and so fails to single out a privileged geometry.

\paragraph{The auxiliary metric as a measure.}\label{par:auxiliary}
One might try to bypass the null degeneracy by replacing the Lorentzian norm with a positive-definite inner product: pick some Riemannian metric $h_{ab}$ on $M$ and define $\|D\|_h^2 := h_{ab}D^a D^b$. This is in the spirit of \cite{WeatherallFletcher}'s framework for first-order approximation, where a Riemannian background does work that the Lorentzian metric cannot. Positive-definite norms have no null vectors, so the degeneracy disappears. However, as we will see, the problem lies instead in the dependence of the resulting criterion on the choice of $h$.

 We want to measure the norm of spacetime vectors that live on $(M,g)$, and $g$ is the natural choice. But auxiliary structure is not objectionable in itself. Mathematics is full of constructions that use an arbitrary auxiliary choice to define something that is independent of it. The standard topology on the space of sections of a vector bundle, for instance, is defined via norms built from an auxiliary fibre metric and connection, and the resulting topology is the same whatever the choice. An auxiliary $h$ might  be argued to be equally innocuous here if the verdicts of the similarity criterion did not depend on the choice. But the dependence is ineliminable.

To see that, we can existentially quantify on the choice of $h$, first, 'there exists one', and then 'for all'. Suppose first that $(M,g)$ being $\delta$-similar to $g'$  requires the bound $\|D\|_h < \delta$ to hold on the tube for \emph{some} Riemannian $h$. Rescaling $h$ by a constant $c^2$ rescales every $h$-norm by $c$. So wherever the supremum of $\|D\|_h$ over a given tube of fixed radius is finite, choosing $c$ small enough brings $\|D\|_h$ below any $\delta$ (and any two Riemannian norms are uniformly equivalent over a compact tube). Every spacetime whose tidal mismatch with $g'$ is merely bounded on the tube then counts as $\delta$-similar to $g'$, for every $\delta$ and every such $g'$. Suppose instead that the bound must hold for \emph{all} Riemannian $h$. Rescaling upward now violates any bound wherever $D^a \neq 0$, so the condition forces $D^a = 0$ exactly, and $\delta$ has dropped out. So, 'there exists one $h$' trivialises the criterion, and `for all $h$' turns it into a criterion of exact coincidence rather than of approximation. %So a nontrivial criterion requires a particular $h$, fixed in advance.

Suppose now that we take these counterexamples as spurious because we think we can somehow normalise the overall scale of $h$. So fix some particular $h$: I will show that the verdicts depend on that choice.  But again, two auxiliary metrics can disagree about \emph{comparative} verdicts. Suppose the difference vectors of two laboratory situations, $D_1$ and $D_2$, point along different directions, with $D_1$ twice as long as $D_2$ in the norm of $h$. A second metric $\tilde{h}$ that stretches the direction of $D_2$ fourfold reverses the ordering. More generally: two auxiliary metrics that differ anisotropically can disagree about which of two situations is closer in tidal action to $g'$, irrespective of mere rescalings. Thus the dependence on $h$ is vicious in an important sense: at fixed $\delta$ and $\rho$, and even merely comparatively, different choices of $h$ can deliver different verdicts about $\delta$-similarity.

Nonetheless, some $h$-independent content does remain. For every Riemannian $h$, the limit $\|D\|_h \to 0$ forces $D^a = 0$. The singular limits of all the $h$-criteria therefore agree, just as the auxiliary choices in the vector-bundle construction agree about the topology.\footnote{Allowing the tube radius $\rho$ to vary with $h$ doesn't change anything: uniform equivalence of Riemannian norms on compact tubes means that the $h$-criteria agree about which quantities vanish in the limit, and that is all.} The singular limit is exact coincidence of tidal acceleration, but $\delta$-flatness is useful precisely away from that limit. Every spacetime is $\delta$-flat for some tube radius $\rho$, and the role of the criterion is to grade approximate flatness at finite $\delta$, telling us which laboratories, and at which scales, may treat their physics as special-relativistic to within $\delta$. So the $h$-criteria agree with one another only in a regime where no interesting relative approximation is left.

I should emphasise that the Lorentzian norm $\|\cdot\|_g$ is not on a par with these choices. There is exactly one physical metric. And the number $\|\text{Dev}(g)\|_g$ has units and can be operationally read as proper distance per proper time squared: the magnitude of tidal acceleration experienced (speaking loosely) by  the inertial observer. In plain terms: clocks and rods measure this magnitude of tidal acceleration, whereas they do not measure its $h$-magnitude.\footnote{This operational contrast presupposes the chronogeometric reading of $g$ flagged in the introduction.}

\paragraph{A remark on the equivalence principle.}\label{par:equiv}
One might invoke the Equivalence Principle here in the following way: for any point $p$ and geodesic $\gamma$, a local diffeomorphism can have $g$ and $g'$ agree to first order (matching connection coefficients), so that the $\delta$-bound is satisfied for sufficiently small deviation vectors. But of course matching connections does not match curvatures: the Riemann tensor depends on \emph{second} derivatives of the metric, so even when $\Gamma \approx \Gamma'$, generally $R_{abcd} \neq R'_{abcd}$. The tidal forces differ by $\mathcal{O}(1)$, not $\mathcal{O}(\delta)$. It is true that tidal acceleration scales linearly with $r$, so shrinking the laboratory would eventually overcome any curvature mismatch. But the $\delta$-flatness criterion takes this into account: it fixes the tube radius $\rho > 0$ in advance and imposes the bound throughout the tube, while the Equivalence Principle guarantees agreement only at a point, to first order. %The contrast shows up in the singular limit. At fixed $\rho$, the bound \eqref{eq:delta_flat} survives $\delta \to 0$ exactly when curvature vanishes throughout the tube, a condition that some spacetimes satisfy. The comparison bound \eqref{eq:delta_sim} with a curved reference survives $\delta \to 0$ at fixed $\rho$ only if the tidal action of $g$ matches that of $g'$ exactly throughout the tube, and first-order matching at a point provides no route to exact matching on a tube. The point here is not the trivial one, that Minkowski happens to be globally flat. Rather, a fixed-radius criterion discriminates among geometries only through conditions that can hold exactly on the tube, and the Equivalence Principle provides no such condition.

\section[Modifying the Criterion: delta-MSS]{Modifying the Criterion: \texorpdfstring{$\delta$}{delta}-MSS}\label{sec:homogeneous}

Section~\ref{sec:obstruction} closed off comparisons with any fixed reference metric that fails to share the lightcones of $g$. What remains is the family of criterion modifications that I am surveying in this paper: inequalities built from the geometry of the physical spacetime alone. (The conformal construction, which escapes the obstruction, is taken up at the end of this section, since its analysis uses the results developed here.) The class of such modifications of my criterion has no principled enclosure. One could bound a curvature scalar rather than a vectorial deviation, abandon the tubular-neighborhood framing, or replace the apparatus altogether. The space of admissible criteria is clearly unbounded, and so this paper does not attempt to chart it. Here I will only investigate what the minimal extensions of $\delta$-flatness looks like. As I said, with this restricted purview there is one main strategy, and a restricted version of it. 

The main strategy is to restrict attention to inequalities of the form $\|\text{Dev}(g)^a - f^a\|_g < \delta$, with $f^a$ a vector built from the metric. This leaves the left-hand side untouched (the magnitude an inertial observer reads off from two test particles) and changes only the target of comparison. Any other choice departs from $\delta$-flatness's original interpretation.

Within that template, we can take $f^a = \Lambda\, r^a$: a single scalar parameter, with no commitment to any decomposition of the Riemann tensor.  The result is (moving indices down for convenience):
\begin{equation}
\label{eq:delta_mss}
\| \text{Dev}(g)_a - \Lambda\, r_a \|_g=\|\left(R_{abcd}- \Lambda\,(g_{ac}\,g_{bd} - g_{ad}\,g_{bc})\right)\, v^b v^d r^c\| < \delta.
\end{equation}
Since $r^a$ is $g$-orthogonal to $v^a$, so is $\Lambda r^a$. The difference vector lies in the spacelike subspace where the Lorentzian norm is positive-definite, and the null degeneracy of \S\ref{sec:obstruction} does not arise. No fixed reference metric appears: the kinematics $(v,r)$, the orthogonality $v \perp r$, and the norm are all supplied by $g$, and $\Lambda$ is a scalar parameter, that appears as a prefactor of a function of the argument metric $g$. Call this criterion $\delta$-MSS.

In the singular limit $\delta \to 0$, \eqref{eq:delta_mss} forces
\begin{equation}
R^a{}_{bcd}\, v^b v^d r^c = \Lambda\, r^a
\end{equation}
for every $g$-unit-timelike $v$ and every $r$ orthogonal to $v$. By the same fact deployed in \S\ref{sec:mechanism}, that sectional curvature on two-planes determines the Riemann tensor, this forces
\begin{equation}
\label{eq:isotropic-Riem}
R_{abcd} = \Lambda\,(g_{ac}\,g_{bd} - g_{ad}\,g_{bc})
\end{equation}
pointwise: the Riemann tensor is isotropic at every point with sectional curvature $\Lambda$.

Moving from the more restricted to the more general case, with a more general $f^a$ vector, we must still guarantee that it is strictly spacelike and built from the metric. Again, the only way to guarantee this is to have $f^a$ proportional to $r^a$. That is, we can ask the question of whether $\Lambda$ in \eqref{eq:delta_mss} can be promoted to a spacetime function $f(x)$, as a functional of the metric. Of course, orthogonality of $f(x)\,r^a$ to $v^a$ is preserved, so the null degeneracy is again avoided, but that is not enough. The singular limit $\delta \to 0$ would force \eqref{eq:isotropic-Riem} with $f(x)$ in place of $\Lambda$, and for $n \geq 3$ the contracted second Bianchi identity then forces $\nabla_b f = 0$. This is Schur's theorem:\footnote{Due to the geometer Friedrich Schur (1886), not Issai Schur of representation theory.} a connected Lorentzian (or Riemannian) manifold of dimension $n \geq 3$ that is isotropic at every point has constant sectional curvature. So, in order to have a singular limit $f$ must be everywhere constant, and the $\delta \to 0$ saturating geometries of \eqref{eq:delta_mss} are exactly the constant-sectional-curvature spacetimes: Minkowski ($\Lambda = 0$), de Sitter ($\Lambda > 0$), and Anti-de Sitter ($\Lambda < 0$).

\paragraph{Not a comparison.}
$\delta$-MSS might look like the special case of $\delta$-similarity in which $g'$ is itself maximally symmetric. But I think it is a better classified as an alternative criterion, for the following reasons. For an MSS reference $g'_\Lambda$ with sectional curvature $\Lambda$, the algebraic Ref-construction of \S\ref{sec:obstruction} gives
\begin{equation}
\text{Ref}(g'_\Lambda)^a = \Lambda\,\big[\,g'_\Lambda(v,v)\,r^a - g'_\Lambda(v,r)\,v^a\,\big],
\end{equation}
\emph{not} $\Lambda\,r^a$. The $g'_\Lambda$-inner products are evaluated on physical kinematic objects: for instance, $v$ is $g$-unit-timelike, not necessarily $g'_\Lambda$-unit-timelike. For these primed inner products to nonetheless give the required results, $g'_\Lambda$ would have to agree pointwise with $g$ on the relevant contractions. But $g'_\Lambda$ is fixed and $g$ is arbitrary: no single MSS metric can get the right inner-product structure for every physical $g$. The conclusion is that $\delta$-MSS is therefore not a comparison-style criterion but a $g$-internal criterion, a modification of the tidal force criterion with $\Lambda$ as a scalar parameter and the MSS family appearing only as the family of singular-limit solutions, labelled by $\Lambda$.
 Different values of $\Lambda$ therefore define different criteria. 

\paragraph{The conformal criterion.}\label{par:conformal-mod}
One other relatively simple modified criterion escapes the obstruction of \S\ref{sec:obstruction} and, because so simple, can be treated within this paper.\footnote{Suggested in passing by \citet[\S 2.1.2, p.~7]{LinnemannReadPuzzle} as a way around the null degeneracy of \S\ref{sec:obstruction}.} The ideas is to restrict the reference $g'$ to share $g$'s lightcones: $g'_{ab} = \Omega^2 g_{ab}$, with $\Omega(x) > 0$. Then define $\delta_c$-flatness by
\begin{equation}
\label{eq:delta_c}
\|\text{Dev}(g)^a - \text{Dev}(g')^a\|_g < \delta.
\end{equation}
Because the two metrics share their lightcones, both deviation terms are $g$-spacelike and the norm is positive-definite. The null degeneracy does not arise, so the question becomes whether the surviving criterion singles out a privileged non-Minkowski geometry. The construction may look like a comparison with a reference metric, but again, I think that is a bad fit, for the same reason it was a bad fit for the $\delta$-MSS. The criterion is indexed by the scalar $\Omega$, and once $\Omega$ is fixed the comparator $\Omega^2 g$ varies with the physical metric: it does not involve a fixed reference metric. Accordingly, the analysis below and in Appendix~\ref{app:conformal-saturation} holds $\Omega$ fixed and asks which metrics saturate the bound.

The singular limit $\delta\to 0$ splits into three cases. If $\Omega = 1$ identically, $\text{Dev}(g')^a = \text{Dev}(g)^a$ and the condition is empty. If $\Omega \neq 1$ is constant, the Christoffel symbols are invariant (${R'}^a{}_{bcd} = R^a{}_{bcd}$). Renormalising the four-velocity to $v'^a = \Omega^{-1}v^a$ gives $\text{Dev}(g')^a = \Omega^{-2}\text{Dev}(g)^a$, so $\text{Dev}(g)^a - \text{Dev}(g')^a = (1 - \Omega^{-2})\text{Dev}(g)^a$ and saturation demands $\text{Dev}(g)^a = 0$. In this case we have reinvented $\delta$-flatness. On the other hand, if $\Omega(x)$ varies, the analysis branches according to whether the gradient $\nabla\Omega$ is null (Appendix~\ref{app:conformal-saturation}). In the \emph{non-null branch}, integrability conditions combined with the saturation-forced conformal flatness of $g$ force $g$ to be locally maximally symmetric, so the construction collapses to $\delta$-MSS. In the \emph{null branch}, $g$ must be conformally flat with pure-radiation Ricci aligned along $\nabla\Omega$, and explicit local families of saturating metrics exist for fixed $\Omega$. So here uniqueness fails. The conformal comparison therefore either collapses to $\delta$-MSS or fails to single out a unique geometry.

\section{Conclusion}\label{sec:conclusion}

\citet{WeatherallFletcher} argued that there is no substantive sense in which flat spacetime plays a privileged role in General Relativity: to first order, every Lorentzian metric locally approximates every other. \citet{LinnemannReadPuzzle} observe that even more sophisticated defences, such as the coordinate-free bitensor expansions of \citet{Teh2024}, still rely on particular coordinate frames, or on picking some Minkowski metric to expand around. In earlier work \citep{Gomes_flat} I proposed $\delta$-flatness as a way around the worry. The criterion bounds the relative acceleration of neighbouring geodesics, and in the limit $\delta \to 0$ only Minkowski satisfies the inequality. Because the criterion names no flat metric, Linnemann and Read's base-point question never arises.

This paper asked whether $\delta$-flatness extends to other reference geometries. There are two ways one might try to cash such extensions out. The first is to compare the tidal acceleration in $(M,g)$ with the tidal acceleration that would arise in some reference spacetime $(M,g')$, asking whether the difference is small. As we saw, Lorentzian signature blocks this for every reference that doesn't share the lightcones of $g$: the difference between the two tidal forces is not spacelike, so it can have zero norm without vanishing, and infinitely many distinct spacetimes saturate the bound.

The second attempt to extend $\delta$-flatness is to drop reference metrics altogether and build the criterion from the geometry of $g$ alone. The space of such criteria has no principled limit: any inequality that we could write down from $g$ and its curvature tensors would be a candidate, and reverse-engineering the inequality that would single out a given $g'$ looks insurmountable. Thus in this paper I do not even consider surveying this space. I looked only at the modifications closest to $\delta$-flatness, in which one bounds the deviation of tidal acceleration from a target vector built from the kinematic data. The minimal candidate is $\delta$-MSS, with target $\Lambda r^a$, a single constant times the deviation vector. As $\delta \to 0$, the Riemann tensor must be pointwise isotropic. Schur's theorem then implies that the saturating geometries are Minkowski, de Sitter, and Anti-de Sitter, one for each value of $\Lambda$. A more generally conformally rescaled reference also evades the degeneracy but does not deliver a new privileged geometry: it either collapses to $\delta$-MSS or admits families of saturating metrics.

To sum up: I have argued that $\delta$-flatness is a substantive criterion for a generic spacetime to approximate a flat spacetime. For it, nothing is put in by hand, and Minkowski emerges as the unique spacetime that saturates the inequality bound. $\delta$-MSS is weaker: one must supply $\Lambda$ independently of the physical metric, and what is then picked out by saturating the bound is a maximally symmetric spacetime. Other criteria, built from different physical and mathematical tools, might pick out other geometries, but constructing such a criterion would be a separate project, which would also require its own motivations.

A natural next step is to recast this rigidity in the language of Cartan geometries. Maximally symmetric spacetimes are the homogeneous model spaces $G/H$ underlying Cartan geometries, and the $\delta$-criterion might be read as a bound on Cartan curvature measuring departure from the model \citep{Sharpe2000,Wise2009}.

\subsection*{Acknowledgments}
I would like to thank James Read, Oliver Pooley, Niels Linnemann, and Jim Weatherall for many conversations on this topic. The author gratefully acknowledges funding from the European Research Council, Grant 101088528 COGY.

\appendix

\section{Rigidity and Non-Uniqueness for the Conformal Saturation Condition}
\label{app:conformal-saturation}

This appendix collects the mathematical results used in the analysis of the conformal comparison (\S\ref{sec:homogeneous}). The aim is to prove two claims:

\begin{enumerate}
\item In the non-null-gradient branch, the conformal saturation condition forces the spacetime to be locally maximally symmetric.
\item In the null-gradient branch, the same condition admits local families of saturating metrics (for fixed metric-independent conformal factor), so no general uniqueness result is available.
\end{enumerate}

Throughout, $d=4$, the signature is Lorentzian, and indices are lowered/raised using the physical metric $g$. We write $g'_{ab}=e^{2\phi}g_{ab}$, $\Omega=e^\phi>0$.

\subsection{Setup and conventions}
\label{app:conformal-setup}

Let $R_{abcd}$ and $R'_{abcd}$ denote the Riemann tensors of $g$ and $g'$, respectively, with all indices lowered using $g$. The conformal saturation condition is
\begin{equation}
\label{eq:app-saturation}
R'_{abcd}=e^{2\phi}R_{abcd}.
\end{equation}

In $4$ dimensions, decompose the Riemann tensor as
\begin{equation}
\label{eq:app-riemann-decomp}
R_{abcd}=W_{abcd}+(g\owedge P)_{abcd},
\end{equation}
where $W_{abcd}$ is the Weyl tensor, $P_{ab}$ is the Schouten tensor,
\[
P_{ab}:=\frac12\left(R_{ab}-\frac16 R\,g_{ab}\right),
\]
and $(g\owedge P)_{abcd}$ is the Kulkarni--Nomizu product.

The conformal transformation law for the Schouten tensor is:
\begin{equation}
\label{eq:app-schouten-transform}
P'_{ab}
=
P_{ab}
-\nabla_a\nabla_b\phi
+\nabla_a\phi\,\nabla_b\phi
-\frac12(\nabla\phi)^2g_{ab}.
\end{equation}

\subsection{First reduction: Weyl vanishing on the nontrivial region}
\label{app:weyl-reduction}

Define the \emph{nontrivial region} $U:=\{x\in M: e^{2\phi(x)}\neq 1\}$.

\begin{lemma}[Weyl vanishing on $U$]
\label{lem:app-weyl}
If \eqref{eq:app-saturation} holds on an open set, then on $U$, $W_{abcd}=0$. Hence $g$ is conformally flat on $U$.
\end{lemma}

\begin{proof}
Under $g'_{ab}=e^{2\phi}g_{ab}$, the Weyl tensor with all indices lowered transforms as $W'_{abcd}=e^{2\phi}W_{abcd}$. Decomposing both sides of \eqref{eq:app-saturation} and comparing the Weyl parts yields $(1-e^{2\phi})W_{abcd}=0$. Thus $W_{abcd}=0$ wherever $e^{2\phi}\neq 1$.
\end{proof}

\subsection{Second reduction: the Schouten PDE}
\label{app:schouten-pde}

On $U$, Lemma~\ref{lem:app-weyl} gives $W=0$. Hence the saturation condition reduces to a condition on the Schouten tensor.

\begin{lemma}[Schouten PDE]
\label{lem:app-schouten-pde}
On $U$, \eqref{eq:app-saturation} implies $P'_{ab}=e^{2\phi}P_{ab}$, and therefore
\begin{equation}
\label{eq:app-schouten-pde}
\nabla_a\nabla_b\phi
=
(1-e^{2\phi})P_{ab}
+\nabla_a\phi\,\nabla_b\phi
-\frac12(\nabla\phi)^2g_{ab}.
\end{equation}
\end{lemma}

Introduce the notation $u_a:=\nabla_a\phi$, $u^2:=u_a u^a$, $\lambda:=1-e^{2\phi}$. Then \eqref{eq:app-schouten-pde} becomes
\begin{equation}
\label{eq:app-u-pde}
\nabla_a u_b=\lambda P_{ab}+u_a u_b-\frac12 u^2 g_{ab}.
\end{equation}

\subsection{Integrability of the Schouten PDE}
\label{app:integrability}

The key rigidity comes from an integrability condition for \eqref{eq:app-u-pde}.

\begin{prop}[Algebraic integrability identity]
\label{prop:app-integrability}
Assume \eqref{eq:app-saturation} on an open set, and restrict to $U$. Then
\begin{equation}
\label{eq:app-key-algebraic}
u_aP_{cb}-u_cP_{ab}-g_{cb}P_{ae}u^e+g_{ab}P_{ce}u^e=0.
\end{equation}
\end{prop}

\begin{proof}
Antisymmetrize $\nabla_c(\cdot)$ and $\nabla_a(\cdot)$ in \eqref{eq:app-u-pde} and use the Ricci identity $(\nabla_c\nabla_a-\nabla_a\nabla_c)u_b = R_{cab}{}^{d}u_d$. On $U$, $W=0$ implies the Cotton tensor vanishes: $\nabla_cP_{ab}-\nabla_aP_{cb}=0$. Using the Schouten form of the Riemann tensor and simplifying yields \eqref{eq:app-key-algebraic}.
\end{proof}

\subsection{The non-null branch: local maximal symmetry}
\label{app:nonnull-branch}

\begin{thm}[Rigidity in the non-null branch]
\label{thm:app-nonnull-rigidity}
Let $U_*:=\{x\in U: u^2(x)\neq 0\}$. If \eqref{eq:app-saturation} holds on an open set, then on each connected component of $U_*$, the metric $g$ is locally maximally symmetric.
\end{thm}

\begin{proof}
Define $J:=P^a{}_a$ and the tracefree Schouten tensor $S_{ab}:=P_{ab}-\frac{J}{4}g_{ab}$. Contracting \eqref{eq:app-key-algebraic} with $g^{ab}$ yields $P_{ab}u^b=\frac{J}{4}u_a$, equivalently $S_{ab}u^b=0$. Substituting into \eqref{eq:app-key-algebraic} yields $u_{[a}S_{c]b}=0$.

At any point in $U_*$, $u_a$ is non-null. For each fixed $b$, the 1-form $S_{\cdot b}$ is proportional to $u_\cdot$, so $S_{ab}=u_a\xi_b$ for some $\xi_b$. Symmetry gives $\xi_a=\sigma u_a$, hence $S_{ab}=\sigma u_a u_b$. Taking the trace: $0=\sigma u^2$. On $U_*$, $u^2\neq 0$, so $\sigma=0$ and $S_{ab}=0$.

Thus $P_{ab}=\frac{J}{4}g_{ab}$: $g$ is Einstein. Combined with $W=0$ from Lemma~\ref{lem:app-weyl}, $g$ has constant sectional curvature.
\end{proof}

\subsection{The null branch: algebraic characterization}
\label{app:null-branch}

\begin{prop}[Structure of the null-gradient branch]
\label{prop:app-null-structure}
Assume \eqref{eq:app-saturation} on an open set, and restrict to $U_{\mathrm{null}}:=\{x\in U: u^2=0\}$. Then:
\begin{enumerate}
\item $W_{abcd}=0$ (conformal flatness).
\item $R=0$.
\item The Schouten tensor is pure radiation aligned with $u_a$: $P_{ab}=\sigma\,u_a u_b$ for some scalar field $\sigma$. Equivalently, $R_{ab}=2\sigma\,u_a u_b$.
\item The Hessian of $\phi$ is rank one: $\nabla_a\nabla_b\phi=\alpha\,u_a u_b$ for $\alpha=1+(1-e^{2\phi})\sigma$.
\end{enumerate}
\end{prop}

\begin{proof}
Item (1) is Lemma~\ref{lem:app-weyl}. When $u$ is null, \eqref{eq:app-key-algebraic} allows $S_{ab}=\sigma u_a u_b$. Differentiating $u^2=0$ and using \eqref{eq:app-u-pde} yields $P_{ab}u^b=0$. Combined with $P_{ab}u^b=\frac{J}{4}u_a$, this forces $J=0$, hence $R=0$ and $P_{ab}=\sigma u_a u_b$. The Hessian statement follows from substituting into \eqref{eq:app-u-pde}.
\end{proof}

\subsection{Local ansatz for the null branch and non-uniqueness}
\label{app:null-ansatz}

We exhibit explicit local families of solutions in the null branch for a \emph{fixed metric-independent} conformal factor.

By Lemma~\ref{lem:app-weyl}, on $U$ the metric is conformally flat. Introduce local double-null coordinates $(u,v,x,y)$ and write
\begin{equation}
\label{eq:app-ansatz-g}
g=e^{2\psi}\eta,\qquad
\eta=-2\,du\,dv+dx^2+dy^2.
\end{equation}
Fix $\phi=\phi(u)$, so $u_a=\phi'(u)\,du_a$ and $u^2=0$. The null-branch characterization forces $\psi=\psi(u)$, and the saturation condition reduces to
\begin{equation}
\label{eq:app-null-ode}
-(\psi+\phi)''+(\psi'+\phi')^2
=
e^{2\phi}\big(-\psi''+(\psi')^2\big).
\end{equation}
Setting $y:=\psi'$, this becomes a Riccati equation
\begin{equation}
\label{eq:app-riccati}
y'
=
y^2
-\frac{2\phi'}{1-e^{2\phi}}\,y
-\frac{(\phi')^2-\phi''}{1-e^{2\phi}}
\end{equation}
with coefficients determined by the fixed function $\phi(u)$.

\begin{prop}[Local non-uniqueness in the null branch]
\label{prop:app-null-nonuniqueness}
Fix a smooth function $\phi(u)$ on an interval $I$ with $e^{2\phi}\neq 1$ on $I$. Then the null-branch saturation condition admits a local family of metrics of the form \eqref{eq:app-ansatz-g}, parametrized by initial data for the Riccati equation \eqref{eq:app-riccati}.
\end{prop}

\begin{proof}
For fixed $\phi$, \eqref{eq:app-riccati} is a smooth first-order ODE. By the standard existence and uniqueness theorems for solutions of ODEs, for any initial value $y(u_0)=y_0$, there exists a unique local solution. Integrating fixes $\psi$ up to a constant; each such $\psi$ defines a local metric satisfying the saturation condition. Finally, varying $y_0$ produces a continuous family of saturating metrics for the same fixed $\Omega=e^\phi$.
\end{proof}

\begin{remark}
The reduced equations are locally well-posed: given initial data, the Riccati equation has a unique solution. But this is not the relevant notion of uniqueness. The saturation condition itself does not fix the initial data. Any choice of $y_0$ yields a metric satisfying the condition, so the condition fails to single out a unique saturating geometry. The criterion was supposed to characterize a privileged reference spacetime but it clearly can only characterise an entire family. 
\end{remark}

\subsection{Summary}
\label{app:summary-branches}

\begin{thm}[Summary]
\label{thm:app-summary}
Let $g'_{ab}=e^{2\phi}g_{ab}$ and assume the conformal saturation condition \eqref{eq:app-saturation}. On the nontrivial region $U=\{e^{2\phi}\neq 1\}$:
\begin{enumerate}
\item $W_{abcd}=0$, so $g$ is conformally flat.
\item In the non-null branch $(\nabla\phi)^2\neq 0$, $g$ is locally maximally symmetric.
\item In the null branch $(\nabla\phi)^2=0$, any solution lies in a conformally-flat pure-radiation class; moreover, for fixed metric-independent $\phi$, explicit local families of saturating metrics exist.
\end{enumerate}
Therefore the conformal saturation condition does not, in general, uniquely determine a metric.
\end{thm}

\bibliographystyle{plainnat}
\bibliography{references3}
\end{document}